\documentclass{svproc}
\usepackage{url}

\usepackage{graphicx}
\usepackage{float}
\usepackage{amsmath}
\usepackage{amssymb}
\usepackage{stmaryrd}
\usepackage{hyperref}
\usepackage[capitalise,nameinlink, noabbrev]{cleveref}

\usepackage{tikz}
\usepackage{thm-restate}
\usepackage{enumitem}

\usepackage{xcolor}
\colorlet{myGreen}{green!50!black}
\colorlet{myLightgreen}{green}
\colorlet{myRed}{red!90!black}
\definecolor{myBlue}{rgb}{0.25, 0.0, 1.0}
\definecolor{myLightBlue}{rgb}{0.39, 0.58, 0.93}
\colorlet{myViolet}{myBlue!55!myRed}
\definecolor{myOrange}{rgb}{1.0, 0.66, 0.07}
\definecolor{magenta}{rgb}{0.94, 0.05, 0.53}
\definecolor{AO}{rgb}{0.0, 0.5, 0.0}
\definecolor{phthaloblue}{rgb}{0.0, 0.06, 0.54}
\definecolor{pistachio}{rgb}{0.58, 0.77, 0.45}
\definecolor{darkgoldenrod}{rgb}{0.72, 0.53, 0.04}

\hypersetup{
  colorlinks=true,
  linkcolor=myBlue,
  citecolor=myGreen,
  urlcolor=myGreen,
  bookmarksopen=true,
  bookmarksnumbered,
  bookmarksopenlevel=2,
  bookmarksdepth=3
  pdftitle = {Graph Searches},
  pdfauthor= {Matjaž Krnc, Nevena Pivač},
}

\begin{document}
\mainmatter              
\title{Graphs where Search Methods are Indistinguishable}
\titlerunning{Graphs where Search Methods are Indistinguishable}  
%
\author{Matjaž Krnc \and Nevena Pivač}
\authorrunning{Krnc and Pivač} 
%
\tocauthor{Matjaž Krnc, Nevena Pivač}
\institute{
The Faculty of Mathematics, Natural Sciences and Information Technologies, \\
University of Primorska, Slovenia.\\
}

\maketitle              

\begin{abstract}
Graph searching is one of the simplest and most widely used tools in graph algorithms. Every graph search method is defined using some particular selection rule, and the analysis of  the corresponding vertex orderings can aid greatly in devising algorithms, writing proofs of correctness, or recognition of various graph families.

We study graphs where the sets of vertex orderings produced by two different search methods coincide. We characterise such graph families for ten pairs from the best-known set of graph searches: Breadth First Search (BFS), Depth First Search (DFS),  Lexicographic Breadth First Search (LexBFS) and Lexicographic Depth First Search (LexDFS), and Maximal Neighborhood Search (MNS). 

\end{abstract}

\newif\iflong
\newif\ifshort
\shorttrue

\section{Introduction}
Graph search methods (for instance, Depth First Search and Breadth First Search) are among essential concepts classically taught at the undergraduate level of computer science faculties worldwide.
Various types of graph searches have been  studied since the 19th century, and used to solve diverse problems, from solving mazes, to linear-time recognition of interval graphs, finding minimal path-cover of co-comparability graphs, finding perfect elimination order, or optimal coloring of a chordal graph, and many others  \cite{arikati1990linear,beisegel2018characterising,corneil2013,corneil2016power,golumbic2004book,kohler2014linear,rose1976,tarjan1972depth}. 

In its most general form, a \emph{graph search} (also \emph{generic search} \cite{corneil2008unified}) is a method of traversing vertices of a given graph such that every prefix of the obtained vertex ordering induces a connected graph. 
This general definition of a graph search  leaves much freedom for a selection rule determining which node is chosen next. By defining some specific rule that restricts this choice, various different graph search methods are defined. Other search methods that we focus on in this paper are Breadth First Search, Depth First Search, Lexicographic Breadth First Search, Lexicographic Depth First Search, 
and Maximal Neighborhood Search.  

This paper is structured as follows. 
In  \cref{sec:prelims} we briefly present the studied graph search methods, and then state the obtained results in \cref{sec:contributions}.
In \cref{sec:P4-C4-pan-diamond} we 
provide a short proof of \cref{thmGENdfsBFS}, as it is the easiest to deal with.
Due to lack of space we omit the proofs of \cref{thm:(L)DFS,thm:thmMNSldfsLBFS}, and 
provide some directions for further work in \cref{sec:conclusion}.

\section{Preliminaries}\label{sec:prelims}
We now briefly describe the above-mentioned graph search methods, and give the formal definitions. 
Note that the definitions below are not given in 
a historically standard form, but rather as so-called \emph{three-point conditions}, due to Corneil and Kruger \cite{corneil2008unified} and also Br\" andstadt~et.~al.~\cite{MR1454439}.

\vspace{2.5mm}\noindent\textbf{Breadth First Search} (BFS), first introduced in 1959 by Moore~\cite{moore1959shortest}, 
is a restriction of a generic search which puts unvisited vertices in a queue and visits a first vertex from the queue in the next iteration. After visiting a particular vertex, all its unvisited neighbors are put at the end of the queue, in an arbitrary order.
\begin{definition}
\label{thm:bfs-ordering-characterization}
An ordering $\sigma$ of $V$ is a BFS-ordering if and only if the following holds: if $ a <_\sigma b <_\sigma c$ and $ac \in E$ and $ab \notin E$, then there exists a vertex $d$ such that $d < a$ and $db \in E$. 
\end{definition}
Any BFS ordering of a graph $G$ starting in a vertex $v$ results in a rooted tree (with root $v$), which contains the shortest paths from $v$ to any other vertex in $G$~(see~\cite{even2011graph}). 
We use this property implicitly throughout the paper.

\vspace{2.5mm}\noindent\textbf{Depth First Search} (DFS), in contrast with the BFS, traverses the graph as deeply as possible, visiting a neighbor of the last visited vertex whenever it is possible, and backtracking only when all the neighbors of the last visited vertex are already visited. In DFS, the unvisited vertices are put on top of a stack, so visiting a first vertex in a stack means that we always visit a neighbor of the most recently visited vertex. 
\begin{definition}
\label{thm:dfs-ordering-characterization}
An ordering $\sigma$ of V is a DFS-ordering if and only if the following holds: if $ a <_\sigma b <_\sigma c$ and $ac \in E$ and $ab \notin E$, then there exists a vertex $d$ such that $a<_\sigma d <_\sigma b$ and $db \in E$. 
\end{definition}
The algorithm for DFS has been known since the nineteenth century as a technique for threading mazes, known under the name \emph{Tr\' emaux's algorithm} (see~\cite{lucas1882recreations}).

\vspace{2.5mm}\noindent\textbf{Lexicographic Breadth First Search} (LexBFS) was introduced in the 1970s by Rose, Tarjan and Lueker~\cite{rose1976} as a part of an algorithm for recognizing chordal graphs in linear time.
Since then, it has been used in many graph algorithms mainly for the recognition of various graph classes. 
\begin{definition}
\label{thm:lbfs-ordering-characterization}
An ordering $\sigma$ of V is a LexBFS ordering if and only if the following holds: if $ a <_\sigma b <_\sigma c$ and $ac \in E$ and $ab \notin E$, then there exists a vertex $d$ such that $d<_\sigma a$ and $db \in E$ and $dc\notin E$. 
\end{definition}
LexBFS is a restricted version of Breadth First Search, where the usual queue of vertices is replaced by a queue of unordered subsets of the vertices which is sometimes refined, but never reordered. 

\vspace{2.5mm}\noindent\textbf{Lexicographic Depth First Search} (LexDFS) was introduced in 2008 by Corneil and Krueger \cite{corneil2008unified} and represents a special instance of a Depth First Search. 
\begin{definition}
\label{thm:ldfs-ordering-characterization}
An ordering $\sigma$ of V is a LexDFS ordering if and only if the following holds: if $ a <_\sigma b <_\sigma c$ and $ac \in E$ and $ab \notin E$, then there exists a vertex $d$ such that $a<_\sigma d <_\sigma b$ and $db \in E$ and $dc\notin E$. 
\end{definition}

\vspace{2.5mm}\noindent\textbf{Maximal Neighborhood Search} (MNS), introduced in 2008 by Corneil and Krueger~\cite{corneil2008unified}, is a common generalization of LexBFS, LexDFS, and MCS, and also of Maximal Label Search (see \cite{Berry2009mls} for definition). 
\begin{definition}
\label{thm:mns-ordering-characterization}
An ordering $\sigma$ of $V$ is an MNS ordering if and only if the following
statement holds: If $a <_\sigma  b <_\sigma  c$ and $ac \in E$ and $ab \notin E$, then there exists a vertex $d$ with $d <_\sigma b$ and $db \in E$ and $dc\notin E$.
\end{definition}
The MNS algorithm uses the set of integers as the label, and at every step of iteration chooses the vertex with maximal label under set inclusion. 

Corneil \cite{corneil2008unified} exposed an interesting structural aspect of graph searches: the particular search methods can be seen as restrictions, or special instances of some more general search methods. For six well-known graph search methods they present a depiction, similar to the one in \cref{fig:graphsearchrelations2}, showing how those methods are related under the set inclusion.
For example, every LexBFS ordering is at the same time an instance of BFS and MNS ordering of the same graph. 
Similarly, every LexDFS ordering is at the same time also an instance of MNS, and of DFS~(see~\cref{fig:graphsearchrelations2}).
The reverse, however, is not true, and there exist orderings that are BFS and MNS, but not LexBFS, or that are DFS and MNS but not LexDFS.

\section{Problem description and results}\label{sec:contributions}
 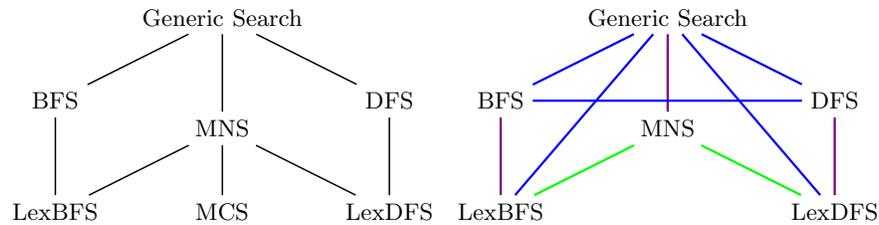
\begin{figure}[b]
\centering
\begin{tikzpicture}[scale=.74]
        \node (gs) at (3,5.5) {Generic Search};%
          \node (bfs) at (0,4){BFS};
          \node (dfs) at (6,4) {DFS};%
          \node (mns) at (3,3.5) {MNS};%
          \node (mcs) at (3,2) {MCS};%
          \node (lbfs) at (0,2) {LexBFS};%
          \node (ldfs) at (6,2) {LexDFS};%
          \draw[line width=0.2mm] (gs)--(bfs);%
          \draw[line width=0.2mm] (gs)--(mns);%
          \draw[line width=0.2mm] (gs)--(dfs);%
          \draw[line width=0.2mm] (bfs)--(lbfs);%
          \draw[line width=0.2mm] (mcs)--(mns);%
          \draw[line width=0.2mm] (dfs)--(ldfs);%
          \draw[line width=0.2mm] (mns)--(lbfs);%
          \draw[line width=0.2mm] (mns)--(ldfs);%

        \node (gs) at (11,5.5) {Generic Search};%
          \node (bfs) at (8,4){BFS};
          \node (dfs) at (14,4) {DFS};%
          \node (mns) at (11,3.5) {MNS};%
          \node (lbfs) at (8,2) {LexBFS};%
          \node (ldfs) at (14,2) {LexDFS};%
          \draw[line width=0.3mm,blue] (gs)--(bfs);%
          \draw[line width=0.3mm,blue] (gs)--(lbfs);%
          \draw[line width=0.3mm,violet] (gs)--(mns);%
          \draw[line width=0.3mm,blue] (gs)--(dfs);%
          \draw[line width=0.3mm,blue] (gs)--(ldfs);%
          \draw[line width=0.3mm,blue] (bfs)--(dfs);%
          \draw[line width=0.3mm,violet] (bfs)--(lbfs);%
          \draw[line width=0.3mm,violet] (dfs)--(ldfs);%
          \draw[line width=0.3mm,green] (mns)--(lbfs);%
          \draw[line width=0.3mm,green] (mns)--(ldfs);%
          
\end{tikzpicture}
\caption{On the left: Hasse diagram showing how graph searches are refinements of one another. 
On the right is a summary of our results: Green pairs are equivalent on $\{P_4,C_4\}$-free graphs.
Violet pairs are equivalent on $\{$pan, diamond$\}$-free graphs.
Blue pairs are equivalent on $\{$paw, diamond, $P_4,C_4\}$-free graphs.
}
\label{fig:graphsearchrelations2}
\end{figure}
Since the connections in \cref{fig:graphsearchrelations2} represent relations of inclusion, it is natural to ask under which conditions on a graph $G$ the particular inclusion holds also in the converse direction. 
More formally, we say that two search methods are \emph{equivalent on a graph $G$} if the sets of vertex orderings produced by both of them are the same. 
We say that two graph search methods are \emph{equivalent on a graph class ${\cal G}$} if they are equivalent on every member $G\in {\cal G}$. 
Perhaps surprisingly, three different graph families suffice to describe graph classes equivalent for the ten pairs of graph search methods that we consider.
Those are described in \cref{thmGENdfsBFS,thm:(L)DFS,thm:thmMNSldfsLBFS}
below, but first we need a few more definitions.

All the graphs considered in the paper are finite and connected.
A \emph{$k$-pan} is a graph consisting of a $k$-cycle, with a pendant vertex added to it. We say that a graph is \emph{pan-free} if it does not contain a pan of any size as an induced subgraph. A $3$-pan is also known as a \emph{paw graph}.

\begin{restatable}{theorem}{thmGENdfsBFS}
\label{thmGENdfsBFS}Let $G$ be a connected graph. Then the following is equivalent: 
\begin{enumerate}[nosep, labelindent=\parindent,leftmargin=*, label = A\arabic*.]
    \item \label{it:gen4}Graph $G$ is $\{P_4, C_4$, paw, diamond$\}$-free.
    \item \label{it:gen1}Every graph search of $G$ is a DFS ordering of $G$.
    \item \label{it:gen2}Every graph search of $G$ is a BFS ordering of $G$.
    \item \label{it:bfsDfs}Any vertex-order of $G$ is a BFS, if and only if it is a DFS.
    
 \end{enumerate}
\end{restatable}

\begin{restatable}{theorem}{thmDFSLDFS}
\label{thm:(L)DFS}
Let $G$ be a connected graph. Then the following is equivalent:  
\begin{enumerate}[nosep, labelindent=\parindent,leftmargin=*, label = B\arabic*.]
    \item Graph $G$ is $\{$pan, diamond$\}$-free.
    \item Every DFS ordering of $G$ is a LexDFS ordering of $G$.
    \item Every BFS ordering of $G$ is a LexBFS ordering of $G$.
    \item Every graph search of $G$ is an MNS ordering of $G$. 
    
\end{enumerate}
\end{restatable}

\begin{restatable}{theorem}{thmMNSldfsLBFS}
\label{thm:thmMNSldfsLBFS}
Let $G$ be a connected graph. Then the following is equivalent:  
\begin{enumerate}[nosep, labelindent=\parindent,leftmargin=*, label = C\arabic*.]
    \item \label{it:P4C4free}Graph $G$ is $\{P_4,C_4\}$-free.
    \item \label{it:mnsLDFS}Every MNS ordering of $G$ is a LexDFS ordering of $G$.
    \item \label{it:mnsLBFS}Every MNS ordering of $G$ is a LexBFS ordering of $G$.
\end{enumerate}
\end{restatable}

From \cref{thmGENdfsBFS,thm:(L)DFS} we can immediately derive similar results for two additional pairs of graph search methods.

\begin{corollary}
\label{thmGENldfsLBFS}
Let $G$ be a connected graph. Then the following is equivalent: 
\begin{enumerate}[nosep, labelindent=\parindent,leftmargin=*]
    \item[A1.] Graph $G$ is $\{P_4, C_4$, paw, diamond$\}$-free.
    \item[A5.] \label{it:gen-LDFS}Every graph search of $G$ is a LDFS ordering of $G$.
    \item[A6.] \label{it:gen-LBFS}Every graph search of $G$ is a LBFS ordering of $G$.
 \end{enumerate}
\end{corollary}
\section{Proof of \cref{thmGENdfsBFS}}\label{sec:P4-C4-pan-diamond}
The following lemma investigates the case when an input graph contains an induced subgraph from $\{P_4, C_4$, paw, diamond$\}$.

\begin{lemma}\label{lem:generic}
Suppose either of the following is true:
\begin{enumerate}[nosep, labelindent=\parindent,leftmargin=*]
    \item every graph search of $G$ is also a BFS, or
    \item every graph search of $G$ is also a DFS, or
    \item a vertex-order of $G$ is a BFS, if and only if it is a DFS.
\end{enumerate}
Then $G$ is a $\{P_4, C_4$, paw, diamond$\}$-free graph. 
\end{lemma}
\begin{figure}[t]
\centering
\begin{tikzpicture}[scale=0.7, vertex/.style={inner sep=2pt,draw, fill,circle}, ]
\begin{scope}
\node[vertex, label=below:$ a $] (a) at (0,0) {};
\node[vertex, label=below:$ b $] (b) at (1,0) {};
\node[vertex, label=below:$ c $] (c) at (2,0) {};
\node[vertex, label=below:$ d $] (d) at (3,0) {};
\draw[] (a) -- (b) -- (c) -- (d);
\end{scope}
\begin{scope}[xshift=4.5cm]
\node[vertex, label=below:$ a $] (a) at (0,0) {};
\node[vertex, label=below:$ b $] (b) at (2,0) {};
\node[vertex, label=above:$ d $] (d) at (0,1) {};
\node[vertex, label=above:$ c $] (c) at (2,1) {};
\draw[] (d) -- (a) -- (b) -- (c) -- (d);
\end{scope}
\begin{scope}[xshift=12cm]
\node[vertex, label=below:$ d $] (b) at (0,0) {};
\node[vertex, label=above:$ b $] (c) at (1,0.5) {};
\node[vertex, label=above:$ c $] (a) at (0,1) {};
\node[vertex, label=right:$ a $] (d) at (2,0.5) {};
\draw[] (c) -- (a) -- (b) -- (c) -- (d);
\end{scope}
\begin{scope}[xshift=8cm]
\node[vertex, label=above:$ a $] (a) at (1,1.25) {};
\node[vertex, label=below:$ b $] (b) at (0,0.5) {};
\node[vertex, label=below:$ d $] (d) at (2,0.5) {};
\node[vertex, label=below:$ c $] (c) at (1,-0.25) {};
\draw[]   (b) -- (c)--(d)  -- (a) -- (b) -- (d);
\end{scope}
\end{tikzpicture}
\caption{In the examples above, ordering $(c,b,a,d)$ is not BFS, while ordering $(b,c,a,d)$ is not DFS. In the two rightmost examples above, ordering $(c,b,a,d)$ is not MNS. }
\label{fig:BFSnotLBFS}
\end{figure}
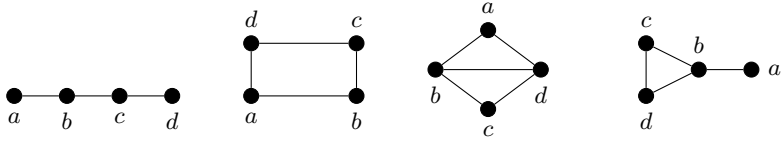
\begin{proof}
Suppose that $G$ contains an induced copy of a graph from $\{P_4, C_4$, paw, diamond$\}$. In other words, $G$ admits a subgraph $H$, where  
$V(H)=\{a,b,c,d\}$ and $\{ab,bc,cd\}\subseteq E(G)$ and $ac\notin E(G)$.
We derive the negations for the three items from this claim.
\begin{enumerate}
    \item Consider any generic search order of $G$ starting with $(c,b,a,\dots)$. Observe that such a vertex-order violates the BFS search paradigm (see \cref{thm:bfs-ordering-characterization}) with the triplet $(c,a,d)$.
    \item Now consider any generic search order of $G$ starting with $(b,c,a,\dots)$. In this case observe that the prefix $(b,c,a)$ of any such vertex-ordering violates \cref{thm:dfs-ordering-characterization}.
    \item It is enough to find a vertex-ordering which is exactly of one among types \{BFS, DFS\}. To this end consider again any search order of $G$ starting with $(c,b,a)$, and continuing so that DFS search order is respected. Similarly as in the item (1) notice that this search again violates the BFS search paradigm (see \cref{thm:bfs-ordering-characterization}), 
    with the triplet $(c,a,d)$.
\end{enumerate}
\end{proof}
We proceed with the proof of the main claim of this section.
\thmGENdfsBFS*
\begin{proof}
By \cref{lem:generic} it is clear that 
\cref{it:gen4} follows independently from either 
\cref{it:gen1}, \ref{it:gen2}, or \ref{it:bfsDfs}
 
We now establish that $G$ is $\{P_4, C_4$, paw, diamond$\}$-free, if and only if it is a star, or a clique. The converse direction is trivial, as 
every star, as well as $K_3$, are $\{P_4, C_4$, paw, diamond$\}$-free. 
For the forward direction  assume that $G$ is a $\{P_4, C_4$, paw, diamond$\}$-free connected graph. We distinguish two cases:
\begin{enumerate}
    \item Graph $G$ is triangle-free. Since it is also 
    $\{P_4, C_4\}$-free, $G$ must be a tree of diameter at most two, which exactly corresponds with the family of stars.
    \item Maximal clique $C$ in $G$ is of size at least three. If $G$ itself is a clique we are done, so suppose that there exists an additional vertex $a\notin C$, such that $N(a)\cap C\neq \emptyset$. Let $b\in N(a)\cap C$ and let $c\in C$ be such that $ac\notin E(G)$ (such a vertex $c$ exists by the maximality of $C$). Finally, since the $C$ is of size at least three, let $d\in C\setminus \{b,c\}$ be an arbitrary remaining vertex of $C$. It remains to observe that $(a,b,c,d)$ induce a paw, or a diamond. 
\end{enumerate}
To conclude the proof, it remains to show that every generic graph search in a clique or a star is also (both) a BFS as well as DFS search. Since in the clique all vertex-orderings are isomorphic, we only consider the case of stars. 
However, observe that stars only admit two non-isomorphic generic vertex orderings, namely the one starting in the center, and the one starting in a leaf. Since both of those vertex-orderings are at the same time also BFS and DFS orders, this concludes the proof of the claim. 
\end{proof}

\section{Conclusion and further work}\label{sec:conclusion}

In this paper we consider the major graph search methods and study the graphs in which vertex-orders of one type coincide with vertex-orders of some other type. 
Interestingly, three different graph families suffice to describe graph classes equivalent for the ten pairs of graph search methods that we consider, which provides an additional aspect of similarities between the studied search methods.

Among the natural graph search methods not yet considered in this setting would be the
\emph{Maximum Cardinality Search} (MCS), introduced in 1984 (for definition see Tarjan and Yannakakis \cite{tarjan1984simple}). As shown on \cref{fig:graphsearchrelations2}, every MCS is a special case of an MNS vertex-order.
While it is easy to verify that  $\{P_4,C_4,\text{paw, diamond}\}$-free graphs do not distinguish between MNS and MCS vertex orders,  \cref{fig:forbiddengraphs-MNS-MCS} provides examples of graphs which admit MNS, but not MNS vertex orders. 
Characterising graphs equivalent for MNS and MCS remains an open question.

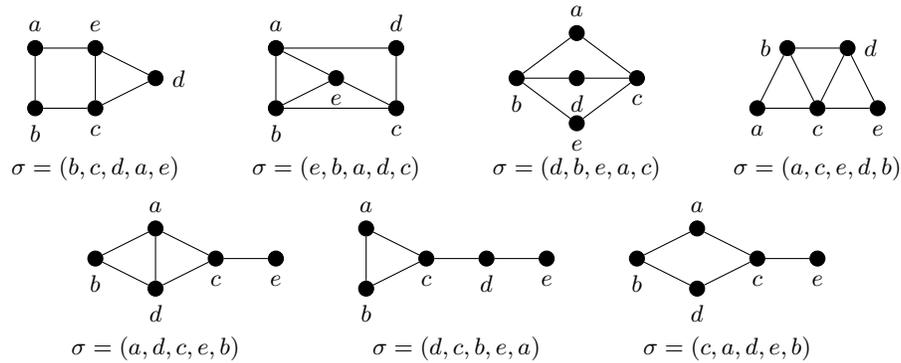
\begin{figure}[H]
\centering
\begin{tikzpicture}[scale=.8,vertex/.style={inner sep=2pt,draw, fill,circle}, ]
\begin{scope}
\node[vertex, label=below:$ b $] (b) at (0,0) {};
\node[vertex, label=below:$ c $] (c) at (1,0) {};
\node[vertex, label=above:$ a $] (a) at (0,1) {};
\node[vertex, label=above:$ e $] (e) at (1,1) {};
\node[vertex, label=right:$ d $] (d) at (2,0.5) {};
\draw[] (a) -- (b) -- (c) -- (d) -- (e) -- (a);
\node[] (s) at (1,-1) {$\sigma=(b,c,d,a,e)$};
\draw[] (e) -- (c);
\end{scope}
\begin{scope}[xshift=5.5cm, yshift=-3cm]
\node[vertex, label=above:$ a $] (a) at (0,1) {};
\node[vertex, label=below:$ b $] (b) at (0,0) {};
\node[vertex, label=below:$ c $] (c) at (1,0.5) {};
\node[vertex, label=below:$ d $] (d) at (2,0.5) {};
\node[vertex, label=below:$ e $] (e) at (3,0.5) {};
\draw[] (c) -- (a) -- (b) -- (c) -- (d) -- (e);
\node[] (s) at (1.5,-1) {$\sigma=(d,c,b,e,a)$};
\end{scope}
\begin{scope}[xshift=10cm, yshift=-3cm]
\node[vertex, label=above:$ a $] (a) at (1,1) {};
\node[vertex, label=below:$ b $] (b) at (0,0.5) {};
\node[vertex, label=below:$ c $] (c) at (2,0.5) {};
\node[vertex, label=below:$ d $] (d) at (1,0) {};
\node[vertex, label=below:$ e $] (e) at (3,0.5) {};
\node[] (s) at (1.5,-1) {$\sigma=(c,a,d,e,b)$};
\draw[] (c) -- (a) -- (b) -- (d)--(c)  -- (e);
\end{scope}
\begin{scope}[xshift=1cm, yshift=-3cm]
\node[vertex, label=above:$ a $] (a) at (1,1) {};
\node[vertex, label=below:$ b $] (b) at (0,0.5) {};
\node[vertex, label=below:$ c $] (c) at (2,0.5) {};
\node[vertex, label=below:$ d $] (d) at (1,0) {};
\node[vertex, label=below:$ e $] (e) at (3,0.5) {};
\node[] (s) at (1,-1) {$\sigma=(a,d,c,e,b)$};
\draw[] (c) -- (a) -- (b) -- (d)--(c)  -- (e);
\draw[] (d) -- (a);
\end{scope}\begin{scope}[xshift=4cm]
\node[vertex, label=above:$ a $] (a) at (0,1) {};
\node[vertex, label=below:$ b $] (b) at (0,0) {};
\node[vertex, label=below:$ c $] (c) at (2,0) {};
\node[vertex, label=above:$ d $] (d) at (2,1) {};
\node[vertex, label=below:$ e $] (e) at (1,0.5) {};
\node[] (s) at (1,-1) {$\sigma=(e,b,a,d,c)$};
\draw[] (b)--(e)--(a) -- (b) -- (c)--(d)--(a);
\draw[] (c) -- (e);
\end{scope}
\begin{scope}[xshift=8cm]
\node[vertex, label=above:$ a $] (a) at (1,1.25) {};
\node[vertex, label=below:$ b $] (b) at (0,0.5) {};
\node[vertex, label=below:$ c $] (c) at (2,0.5) {};
\node[vertex, label=below:$ d $] (d) at (1,0.5) {};
\node[vertex, label=below:$ e $] (e) at (1,-0.25) {};
\node[] (s) at (1,-1) {$\sigma=(d,b,e,a,c)$};
\draw[] (c) -- (d) -- (b) -- (e)--(c)  -- (a)--(b);
\end{scope}
\begin{scope}[xshift=12cm]
\node[vertex, label=below:$ a $] (a) at (0,0) {};
\node[vertex, label=left:$ b $] (b) at (0.5,1) {};
\node[vertex, label=below:$ c $] (c) at (1,0) {};
\node[vertex, label=right:$ d $] (d) at (1.5,1) {};
\node[vertex, label=below:$ e $] (e) at (2,0) {};
\node[] (s) at (1,-1) {$\sigma=(a,c,e,d,b)$};
\draw[] (b)--(c) -- (d) -- (b) -- (a)--(c)  -- (e)--(d);
\end{scope}
\end{tikzpicture}\caption[Orderings that are MNS and not MCS.]{Graphs and corresponding orderings that are MNS and not MCS orderings. }
\label{fig:forbiddengraphs-MNS-MCS}
\end{figure}


\subsection*{Acknowledgements}
The authors would like to thank prof. Martin Milanič for the initial suggestion of the problem, and to Ekki Köhler and his reseach group, for introducing the diverse world of graph searches to us.

\bibliographystyle{bibtex/splncs03}
\bibliography{biblio}

\begin{thebibliography}{10}
\providecommand{\url}[1]{\texttt{#1}}
\providecommand{\urlprefix}{URL }

\bibitem{arikati1990linear}
Arikati, S.R., Rangan, C.P.: Linear algorithm for optimal path cover problem on
  interval graphs. Information Processing Letters  35(3),  149--153 (1990)

\bibitem{beisegel2018characterising}
Beisegel, J.: Characterising {AT}-free graphs with {BFS}. In: Brandst{\"{a}}dt,
  A., K{\"{o}}hler, E., Meer, K. (eds.) Graph-Theoretic Concepts in Computer
  Science. pp. 15--26 (2018)

\bibitem{Berry2009mls}
Berry, A., Krueger, R., Simonet, G.: Maximal label search algorithms to compute
  perfect and minimal elimination orderings. SIAM Journal on Discrete
  Mathematics  23(1),  428--446 (2009)

\bibitem{MR1454439}
Brandst\"{a}dt, A., Dragan, F.F., Nicolai, F.: Lex{BFS}-orderings and powers of
  chordal graphs. Discrete Math.  171(1-3),  27--42 (1997)

\bibitem{corneil2013}
Corneil, D.G., Dalton, B., Habib, M.: {LDFS} based certifying algorithm for the
  {Minimum Path Cover} problem on cocomparability graphs. {SIAM} Journal on
  Computing  42(3),  792--807 (2013)

\bibitem{corneil2016power}
Corneil, D.G., Dusart, J., Habib, M., Kohler, E.: On the power of graph
  searching for cocomparability graphs. SIAM Journal on Discrete Mathematics
  30(1),  569--591 (2016)

\bibitem{corneil2008unified}
Corneil, D.G., Krueger, R.M.: A unified view of graph searching. SIAM Journal
  on Discrete Mathematics  22(4),  1259--1276 (2008)

\bibitem{even2011graph}
Even, S.: Graph algorithms. Cambridge University Press (2011)

\bibitem{golumbic2004book}
Golumbic, M.: Algorithmic Graph Theory and Perfect Graphs, pp. 98--99. Annals
  of Discrete Mathematics, Volume 57, Elsevier (2004)

\bibitem{golumbic1978trivially}
Golumbic, M.C.: Trivially perfect graphs. Discrete Mathematics  24(1),
  105--107 (1978)

\bibitem{kohler2014linear}
K{\"o}hler, E., Mouatadid, L.: Linear time lexdfs on cocomparability graphs.
  In: Scandinavian Workshop on Algorithm Theory. pp. 319--330. Springer (2014)

\bibitem{lucas1882recreations}
Lucas, {\'E}.: R{\'e}cr{\'e}ations math{\'e}matiques: Les traversees. Les
  ponts. Les labyrinthes. Les reines. Le solitaire. la num{\'e}ration. Le
  baguenaudier. Le taquin, vol.~1. Gauthier-Villars et fils (1882)

\bibitem{moore1959shortest}
Moore, E.F.: The shortest path through a maze. In: Proc. Int. Symp. Switching
  Theory, 1959. pp. 285--292 (1959)

\bibitem{olariu1988paw}
Olariu, S.: Paw-free graphs. Information Processing Letters  28(1),  53--54
  (1988)

\bibitem{rose1976}
Rose, D.J., Lueker, G.S., Tarjan, R.E.: Algorithmic aspects of vertex
  elimination on graphs. {SIAM} Journal on Computing  5(2),  266--283 (1976)

\bibitem{tarjan1972depth}
Tarjan, R.E.: Depth-first search and linear graph algorithms. SIAM journal on
  computing  1(2),  146--160 (1972)

\bibitem{tarjan1984simple}
Tarjan, R.E., Yannakakis, M.: Simple linear-time algorithms to test chordality
  of graphs, test acyclicity of hypergraphs, and selectively reduce acyclic
  hypergraphs. SIAM Journal on computing  13(3),  566--579 (1984)

\end{thebibliography}

\ifshort
\newpage
\appendix

\section{Preliminaries}
We denote the $i$-th neighbourhood of a vertex $v$ in $G$ by 
$$
N_G^i(v)=\{w\mid d_G(v,w)=i\}.
$$
We first recall from Olariu~\cite{olariu1988paw}, that the following holds.
\begin{theorem}
\label{thm:olariu}
A paw-free graph is either triangle-free, or complete multipartite. 
\end{theorem}

\section{Proof of \cref{thm:(L)DFS}}\label{sec:pan-diamond}

\subsection{Breadth First Search and Lexicographic Breadth First Search}

Graph search methods in general don't have the hereditary property. Let $G$ be a graph with a search ordering $\sigma$ of particular type, and let $H$ be an induced subgraph of $G$. It is not true that $\sigma^*$ obtained from $\sigma$ by deletion of vertices that are not in $H$ represents a search ordering of the same type of $H$, as can be seen in the following example. 

\begin{example}
Let $G$ be a cycle on $5$ vertices, and let us denote its vertices by $v_1,v_2,v_3,v_4,v_5$ in the cyclic order. It is not difficult to see that $\sigma=(v_1,v_2,v_5,v_3,v_4)$ is a BFS ordering of $G$. Let $H$ be a subgraph of $G$ obtained by deletion of vertex $v_5$, and let $\sigma^*$ be an ordering of vertices in $H$ obtained from $\sigma$ after deletion of $v_5$. Then $\sigma^*=(v_1,v_2,v_3,v_4)$ is not a valid BFS ordering of $H$. 
\end{example}

From the above it follows that it could happen that there is an ordering of a graph $H$ that is BFS and not LexBFS ordering, while in a graph $G$ containing $H$ as an induced subgraph it is not necessarily true. It means that the equivalence between BFS and LexBFS in $G$ does not imply the same equivalence or every induced subgraph of $G$. In the following example we can see that a valid LexBFS ordering of $G$ yields an ordering of its subgrsph $H$ that is BFS and not LexBFS. 

\begin{example}
Let $G$ be a graph from \cref{fig:remove-6-pan}. After removing the vertex $u$ from $G$ we get a $6$-pan $G'$. Observe that in $G'$ we can find a BFS ordering $\sigma^*=(v_1,v_2,v_6,v_3,v_5,v,v_4)$ that is not a valid LexBFS ordering. If $\sigma^*$ is a part of a valid BFS ordering $\sigma$ of $G$, then we must visit $u$ before visiting non-neighbors of $v_1$, and after visiting vertices $v_2$ and $v_6$. Then it follows that $\sigma=(v_1,v_2,v_6,u,v_3,v_5,v,v_4)$ and it represents a valid LexBFS ordering of $G$, so is not an example of ordering of $G$ that is BFS and not LexBFS. 
\end{example}
\begin{figure}[H]
\centering
\begin{tikzpicture}[vertex/.style={inner sep=2pt,draw, fill,circle},xscale=1.5, yscale=1.8]
\begin{scope}
\node[vertex, label=below:$ v_6 $] (v6) at (0,0) {};
\node[vertex, label=below:$ v_5 $] (v5) at (1,0) {};
\node[vertex, label=below:$ v_4 $] (v4) at (2,0) {};
\node[vertex, label=above:$ v_1 $] (v1) at (-0.5,1) {};
\node[vertex, label=above:$ u $] (u) at (0.5,1) {};
\node[vertex, label=above:$ v $] (v) at (1.5,1) {};
\node[vertex, label=above:$ v_3 $] (v3) at (2.5,1) {};
\node[vertex, label=above:$ v_2 $] (v2) at (1,2) {};
\draw[] (v3)--(v)--(u) --  (v1) -- (v2)--(v3)  -- (v4)--(v5) -- (v6) --(v1);
\end{scope}
\begin{scope}[xshift=4.5cm]
\node[vertex, label=below:$ v_6 $] (v6) at (0,0) {};
\node[vertex, label=below:$ v_5 $] (v5) at (1,0) {};
\node[vertex, label=below:$ v_4 $] (v4) at (2,0) {};
\node[vertex, label=above:$ v_1 $] (v1) at (-0.5,1) {};
\node[vertex, label=above:$ v $] (v) at (1.5,1) {};
\node[vertex, label=above:$ v_3 $] (v3) at (2.5,1) {};
\node[vertex, label=above:$ v_2 $] (v2) at (1,2) {};
\draw[] (v3)--(v);
\draw[] (v1) -- (v2)--(v3)  -- (v4)--(v5) -- (v6) --(v1);
\end{scope}
\end{tikzpicture}
\caption[The orderings of a graph and its induced subgraph.]{The ordering $\sigma=(1,2,6,u,3,5,v,4)$ is a valid LexBFS ordering of $G$ (left), while the ordering $\sigma^*=(1,2,6,3,5,v,4)$ is not a valid LexBFS ordering of $G-u$ (right).}
\label{fig:remove-6-pan}
\end{figure}
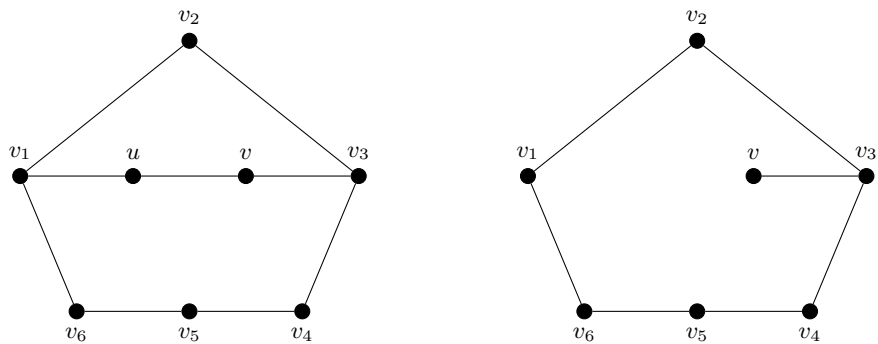

Despite both demotivating examples above, we identify certain graphs where the equivalence between BFS and LexBFS does not hold in any graph containing them as an induced subgraph.

\begin{lemma}\label{lemma:paw-diamond-bfs}
Let $G$ be a graph which contains a diamond or a pan as an induced subgraph. Then there is a BFS ordering of $G$ that is not a LexBFS ordering of $G$.
\end{lemma}
\begin{proof}
First assume that $G$ contains a paw or a diamond as an induced subgraph. We show that there is a BFS ordering of $G$ that is not a LexBFS ordering of $G$ 
The claim can be easily justified by giving a prefix of an order $\sigma$ that is a BFS order and not a LexBFS order of a graph containing a paw or a diamond.
Let $G$ be a graph and let $H$ be a paw graph, contained in $G$ as an induced subgraph. Using the same notation as in~\cref{fig:3BFSnotLBFS} (left) we can define the BFS ordering $\sigma_1$ of $G$ starting in $c$, with first four vertices in $\sigma$ being $c,a,d,b$, in that order. Similarly, if $H$ is a diamond contained in $G$ as induced subgraph, we can define the BFS ordering $\sigma$ of $G$ starting in $c$ and visiting consecutively vertices $b,d,a$ (\cref{fig:3BFSnotLBFS} right). In both cases $\sigma $ is a BFS ordering, since it starts with a vertex $c$ and visits its neighbors. 
Also, $\sigma$ cannot be a LexBFS ordering, since in both cases vertex $a$ has label $\{n, n-1\}$, while $d$ has a label $n$, so $a$ should appear before $d$, no matter how the rest of $\sigma$ is defined. 

Now consider the case when $G$ contains a pan bigger then a paw. So denote $P$ to be a smallest pan in $G$, let $k\ge 4$ be the length of its cycle. Denote vertices of $P$ by $\{v_0,v_1,\dots,v_k\}$ such that $v_k$ is a pendant vertex connected to $v_{\lfloor \frac{k}{2}\rfloor-1}$. 
For any integer $i\in \{1,\dots,\lfloor \frac{k}{2}\rfloor\}$ we first observe the following:
\begin{enumerate}
    \item We have that $\{v_i,v_{k-i}\}\subseteq N_G^i(v_0),\text{ and }v_n\in N_G^{\lfloor k/2\rfloor}(v_0)$.\label{it:panbfs}
    \item Shortest $(v_0v_i)$-path in $G$ is unique and lies in $P$. Similarly, shortest $(v_0v_{k-i})$-path in $G$ is unique and lies in $P$.\label{it:bfs-unique}
    \item Let $P'$ be any shortest path between $v_0$ and a vertex from $N_G^{\lfloor k/2\rfloor}(v_0)$. If $P'$ is not completely contained in $P$, then it does not intersect $P$ (except at endpoints).\label{it:short}
\end{enumerate}

Indeed, any path violating the above would give rise to a pan on less then $k+1$ vertices, contradicting the choice of $P$. 
We distinguish two cases depending on the parity of $k$.

\paragraph{The case where $k$ is odd.}
First observe that for any $i\in \{1,\dots,k-1\}$
the shortest path between $v_0$ and $v_i$ is lying within $P$ and is unique in $G$. 
This is true as an existence of any different shortest path would give rise to a pan smaller then $P$.

Now consider a BFS vertex-ordering $\alpha$ starting at $v_0$, where the first vertex we choose at the distance $i$ from $v_0$ is 
$v_{k-i}$, for any $i\in \{1,2,\dots, (k-1)/2\}$. 
This is always possible as $(v_0,v_{k-1},v_{k-2},\dots,v_{(k+1)/2})$ is a path in $G$.
Moreover, we prioritise choosing a vertex $v_n$ as soon as possible.
By \cref{it:panbfs} we recall that
$\{a_{(k-1)/2},a_{(k+1)/2},a_k\}\subseteq N^{(k-1)/2}(a_0)$.
We next claim that $v_k<_{\alpha}v_{(k-1)/2}$. 
Indeed, the unique shortest path between $v_0$ and $v_{(k-1)/2}$ in $G$ goes through $v_{(k-3)/2}$, which is at the same time adjacent to $v_k$. It is hence always possible to select a vertex $v_k$ before $v_{(k-1)/2}$ in $\alpha$.

Now observe that $v_{(k+1)/2}<_{\alpha}v_k<_{\alpha}v_{(k-1)/2}$, where
$v_{k+1/2}v_{k-1/2}\in E(G)$ while
$v_{(k+1)/2}v_k\notin E(G)$.  
\cref{thm:lbfs-ordering-characterization} hence implies that there exists another vertex $x<_{\alpha}v_{(k+1)/2}$ such that $xv_k\in E(G)$
while $xv_{(k-1)/2}\notin E(G)$.
To this end recall that the first vertex we chose 
from the set $N_G^{(k-1)/2}(v_0)$ was $v_{(k+1)/2}$, so $x<_{\alpha}v_{(k+1)/2}$ implies that 
$d_G(v_0,x)\le (k-3)/2$. 
Let $Q$ be a shortest path between $v_0$ and $x$.
We conclude this case by identifying a pan smaller then $P$,
inside of a graph induced by vertices $\{v_{k-1}, v_k\}\cup\{v_1,\dots,v_{(k-1)/2}\}\cup Q$.

\paragraph{The case where $k$ is even.}
Denote by $\alpha$ a BFS vertex-ordering which starts at $v_0$, and where,
among the eligible vertices, we prioritise vertices from $P$.
As an additional tie-breaking rule we select the vertex from $P$ with the minimal index,
until we have used all vertices at distance at most $k/2-1$ from $v_0$.
Immediately after vertices from $N^{k/2-1}(v_0)$, we append $v_n$, and then $v_{k/2}$ to $\alpha$. 

In particular, 
by  construction of $\alpha$ and by \cref{it:panbfs,it:bfs-unique,it:short} the sequence $\alpha$ must contain the following subsequence
\begin{align}
    v_0 <_\alpha 
    v_1 <_\alpha 
    v_{k-1} <_\alpha 
    v_2 <_\alpha 
    v_{k-2} <_\alpha 
    \dots <_\alpha 
    v_{(k/2)-1} <_\alpha 
    v_{(k/2)+1} <_\alpha 
    v_k <_\alpha 
    v_{k/2}.
\end{align}
Now consider the labels of $v_k$ and $v_{k/2}$ at the moment right before $v_k$ is chosen. Clearly the latter contains the index of vertex $v_{(k/2)+1}$ while the former does not, and clearly both contain the index of vertex $v_{(k/2)-1}$. 
This implies that $\alpha$ is not a LexBFS order as it should chose the vertex $v_{k/2}$ instead of $v_k$.
Here we note that the labels of $v_k$ and $v_{k/2}$ might contain additional entries in its label, however those cannot affect the lexicographic priority of  $v_{k/2}$, as the label of  $v_{(k/2)+1}$ preceeds them all by the definition of $\alpha$, and by \cref{it:short}.
This concludes the proof of the claim.



\end{proof}

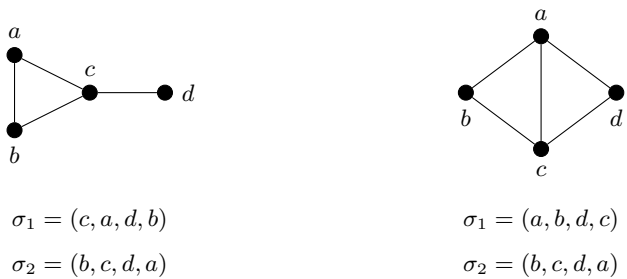
\begin{figure}[H]
\centering
\begin{tikzpicture}[vertex/.style={inner sep=2pt,draw, fill,circle}, ]
\begin{scope}
\node[vertex, label=below:$ b $] (b) at (0,0) {};
\node[vertex, label=above:$ c $] (c) at (1,0.5) {};
\node[vertex, label=above:$ a $] (a) at (0,1) {};
\node[vertex, label=right:$ d $] (d) at (2,0.5) {};
\draw[] (c) -- (a) -- (b) -- (c) -- (d);
\node[] (s) at (1,-1.2) {$\sigma_1=(c,a,d,b)$};
\node[] (s) at (1,-1.8) {$\sigma_2=(b,c,d,a)$};
\end{scope}
\begin{scope}[xshift=6cm]
\node[vertex, label=above:$ a $] (a) at (1,1.25) {};
\node[vertex, label=below:$ b $] (b) at (0,0.5) {};
\node[vertex, label=below:$ d $] (d) at (2,0.5) {};
\node[vertex, label=below:$ c $] (c) at (1,-0.25) {};
\node[] (s) at (1,-1.2) {$\sigma_1=(a,b,d,c)$};
\node[] (s) at (1,-1.8) {$\sigma_2=(b,c,d,a)$};
\draw[] (a) --  (b) -- (c)--(d)  -- (a)--(c);
\end{scope}
\end{tikzpicture}
\caption[Some orderings of a paw and a diamond.]{A paw (left) and a diamond (right). The corresponding search orderings $\sigma_1$ ($\sigma_2$)  are BFS and not LexBFS orderings (DFS and not LexDFS orderings, resp.).
}
\label{fig:3BFSnotLBFS}
\end{figure}

\begin{lemma}
\label{lemma:characterization}
If a connected graph $G$ does not contain a diamond, or a pan as an induced subgraph, then $G$ is either acyclic, or a cycle on at least $4$ vertices, or a complete graph, or a complete bipartite graph.
\end{lemma}

\begin{proof}
Let $G$ be a graph that does not contain a diamond, or a pan as induced subgraph. From~\cref{thm:olariu} it follows that $G$ is either a complete multipartite graph, or a triangle-free graph. 

Let first $G$ be a complete multipartite graph, with partition classes $S_1,\dots, S_k$. If all partition classes of $G$ have one vertex, then $G$ is a complete graph, so we may assume without loss of generality that $|S_1|\ge 2$. Let $x,y\in S_1$. If there are exactly two partition classes of $G$, then $G$ is a complete bipartite graph. Assume that there are at least three partition classes in $G$, and let $z\in S_2$, $w\in S_3$. Then the vertices $\{x,y,z,w\}$ form a diamond in $G$; a contradiction. 

Let now $G$ be a triangle-free graph. If $G$ does not contain any cycle, then $G$ is a tree, and we are done. 
Assume first that $G$ contains a cycle of length at least five and let $C$ be such a cycle in $G$. If $G=C$, we are done, so assume that there is a vertex $v$ in $V(G)\setminus V(C)$ having a neighbor in $C$. If $v$ has exactly one neighbor in $C$, then $V(C)\cup \{v\}$ induce a cycle with pendant vertex in $G$, so $v$ has at least two neighbors in $C$. We know that $G$ is triangle-free, so no two consecutive vertices of $C$ are adjacent to $v$. Let $v_i,v_j\in C$, $i<j$ be neighbors of $v$ such that $|j-i|=j-i$ is minimal. Then $vv_{i-1}\notin E(G)$ and vertices $v, v_{i-1}, v_i,v_{i+1},\dots, v_j$ form a cycle with pendant vertex, unless it holds that $v_{i-1}v_j\in E(G)$, that is, unless the vertices $v_{i-1}$ and $v_j $ are consecutive in $C$, meaning that the distance between $v_i$ and $v_j$ in $C$ is equal to two and that $C$ is a cycle on four vertices. Our assumption was that $C$ is a cycle on at least $5$ vertices, so we have a contradiction. It follows that the vertex $v$ does not exist and $G=C$.

Assume now that any cycle in $G$ contains exactly four vertices, and let $C$ be such a cycle, with vertices $v_1,v_2,v_3,v_4$ in consecutive order. 
We know that $G$ has no odd cycles, so $G$ is bipartite graph. Also, we know that $C$ is a complete bipartite graph. Let $F$ be a subgraph of $F$ that contains $C$ such that $F$ is maximal complete bipartite subgraph of $G$, and let $(A,B)$ be a partition of $F$. Without loss of generality we may assume that $v_1,v_3\in A$ and $v_2,v_4\in B$. If $G=F$, then $G$ is a complete bipartite graph, and we are done, so assume there is a vertex $v\in V(G)\setminus V(F)$. A graph $G$ is connected, so $v$ has a neighbor in $F$. Let without loss of generality $u\in A$ be a neighbor of $v$. We know by definition of $F$ that $u$ is adjacent to all vertices in $B$, so it cannot be that $v$ has a neighbor in $B$, since otherwise that neighbor together with vertices $u$ and $v$ would form a triangle in $G$. It follows that $(A, B\cup \{v\})$ is a partition of a bipartite graph, and from the maximality of $F$ it follows that $v$ has a non-neighbor in $A$. Let $x\in A$ be a non-neighbor of $u$. (Observe that it can happen that $\{x,u\}\cap \{v_1,v_3\}\neq \emptyset$.) Taking the vertices $\{x,u,v_2,v_4,v\}$ we get the forbidden $C_4$ with a pendant edge; a contradiction. It follows that $G=F$ and thus $G$ is a complete bipartite graph, as we wanted to show.
\end{proof}
%


\begin{lemma}\label{lem:(L)BFS-cycles-forests-cliques}
In the following graph classes every BFS ordering is a LexBFS ordering. 
\begin{enumerate}[nosep, label = \roman*)]
\item cycles
\item forests
\item complete graphs
\item complete bipartite graphs
\end{enumerate}
\end{lemma}
\begin{proof}
We prove the lemma for each case separately.

\begin{enumerate}[nosep, label = \roman*)]
\item Assume for contradiction this is not true, and let $G$ be a cycle with ordering $\sigma$ that is a BFS ordering and not a LexBFS ordering. By~\cref{thm:lbfs-ordering-characterization} it follows that there are vertices $a<_\sigma b<_\sigma c$ such that $ab\notin E(G)$, $ac\in E(G)$ and for every $d'<_\sigma a$ it holds that either $d'b\notin E(G)$, or $d'c\in E(G)$. Similarly, from the~\ref{thm:bfs-ordering-characterization} it follows that there is a vertex $d<_\sigma $ such that $db\in E(G)$. Then it must be that $d'c\in E(G)$, so $b$ and $c$ are both neighbors of $d'$ in $G$. We know that $G$ is a cycle, so every vertex in $G$ is of degree $2$, and thus $b$ and $c$ are the only neighbors of $d'$ in $G$. Since $\sigma$ is a BFS ordering, at every step it visits a neighbor of some already visited vertex, so it must be that $\sigma(d')=1$. Then the neighbors of $d'$ are visited before non-neighbors of $d'$, so vertices $b$ and $c$ must be visited before $a$ in the BFS ordering $\sigma$. This is a contradiction with the definition of $a,b,c$, so such an ordering $\sigma$ does not exist, and every BFS ordering of $G$ is also a LexBFS ordering of $G$.
\item Let $\sigma $ be a BFS ordering of a forest graph $G$, and let $\sigma(v)=1$. If we do a LexBFS on $G$ starting in $v$, at every step of iteration all the unvisited vertices have a label consisting just of one number - a number belonging to the parent of the unvisited vertex. Thus, the label of every vertex consists just of a number belonging to the first visited neighbor. It means that putting the vertices in a queue in BFS is exactly the same as ordering vertices with respect to the lexicographic maximal label, so $\sigma$ is a LexBFS of $G$. 

\item If $v$ is arbitrary vertex of a complete graph $G$, once the vertex $v$ is visited, every unvisited vertex in $G$ gets a label from $v$. It means that at iteration step of LexBFS all the unvisited vertices in $G$ have the same label, so we can choose any among them. Any ordering of vertices of a complete graph is BFS and LexBFS ordering. 
\item Let $G$ be a complete bipartite graph with partition classes $A$ and $B$, and let $\sigma$ be a BFS ordering of $G$. Assume without loss of generality that $v\in A$ is a first vertex in ordering $\sigma.$ BFS is a layered search on $G$, so after visiting $v$ we visit all the neighbors of $v$ in $B$. After that, we visit all the vertices that are on distance $2$ from vertex $v$ in $G$, and so on, until we visit all the vertices in $G$. This search is also a LexBFS search, since at every step of LexBFS the vertices of $G$ in the same partition of $V(G)$ all have the same labels, and we can choose any among them. Also, all vertices that are on some distance $i$ from $v$ belong either to $A$ or $B$, so $\sigma$ is a LexBFS ordering of $G$. 
\end{enumerate}
\end{proof}

\cref{lemma:paw-diamond-bfs,lemma:characterization,lem:(L)BFS-cycles-forests-cliques} imply the following.

\begin{corollary}\label{cor:(L)BFS}
For any graph $G$, the following is equivalent: 
\begin{enumerate}[nosep, label = \roman*)]
    \item Every BFS ordering of $G$ is a LexBFS ordering of $G$.
    \item Graph $G$ is $\{$pan, diamond$\}$-free.
\end{enumerate}
\end{corollary}

\subsection{Depth First Search and Lexicographic Depth First Search}

In this section we prove \cref{thm:(L)DFS}.
In the process we utilise the characterization of (Lex)DFS orderings (so-called ``point conditions'') described in \cref{thm:ldfs-ordering-characterization}. 
We start by giving sufficient condition regarding when DFS and LexDFS are not equivalent.
\begin{lemma}
If a graph $G$ contains a pan or a diamond as an induced graph, 
then there is a DFS ordering of $G$ that is not a LexDFS ordering of $G$.
\end{lemma}
\begin{proof}
The claim can be easily justified by giving a prefix of an order $\sigma$ that is a DFS order and not a LexDFS order of a graph containing a pan or a diamond.
First consider the case when $G$ contains a paw as an induced subgraph. Using the same notation as in~\cref{fig:3BFSnotLBFS} (left) we can define the DFS ordering $\sigma_2$ of $G$ starting in $c$, with first four vertices in $\sigma_2$ being $b,c,d,a$, in that order. 

Similarly, if $H$ is a diamond contained in $G$, we can define the DFS ordering $\sigma_2$ of $G$ having the same prefix: starting in $b$ and visiting consecutively vertices $c,d,a$ (\cref{fig:3BFSnotLBFS} right). In both cases $\sigma_2 $ is a DFS ordering, since it starts with a vertex $b$ and traverse the graph as deep as possible.
Also, $\sigma_2$ cannot be a LexDFS ordering, since in both cases the vertex $a$ has a label $\{21\}$, while $d$ has a label $\{1\}$, so $a$ should appear before $d$, no matter how the rest of $\sigma$ is defined. 
\end{proof}

As we already know, paw is defined as $3$-pan. In the following lemma we give a result showing that the eqivalence between DFS and LexDFS in $G$ implies that $G$ does not contain any pan as induced subgraph. This result generalizes the part of previous lemma that considered the existence of a paw graph in $G$. 

\begin{lemma}
If a graph $G$ contains a pan as an induced subgraph, 
then there is a DFS ordering of $G$ that is not a LexDFS ordering of $G$ (that is, DFS and LexDFS are not equivalent in $G$).
\end{lemma}
\begin{proof}
The claim can be easily justified by giving a prefix of an order $\sigma$ that is a DFS order and not a LexDFS order of a graph containing a pan.
Let $G$ be a graph and let $H$ be a pan, contained in $G$ as an induced subgraph. Let the vertices of $H$ be denoted by $v_1,\dots, v_n, v$, where vertices $v_1,v_2,\dots, v_n$ form a cycle in this order, and $v$ is a vertex of degree $1$, adjacent to $v_{n-1}$. 
We can define the DFS ordering $\sigma$ of $G$ starting in $v_1$, with first $n$ vertices in $\sigma$ being $v_1,v_2,\dots, v_{n-2}, v_{n-1},v$, in that order. It is clear that $\sigma $ is a DFS order, since it has a prefix that is a path, and continues traversing the graph $G$ using DFS. At the same time we have that $\sigma$ is not a LexDFS ordering of $G$. We know that the vertex $v_n$ appears in $\sigma$ after all other vertices from $H$. It follows that $v_1<_\sigma v<_\sigma v_n$ with $v_1v_n\in E(G)$ and $v_1v\notin E(G)$. By~\cref{thm:ldfs-ordering-characterization} it follows that there is a vertex $v_i$: $v_1<_\sigma v_i<_\sigma v$, $v_iv\in E(G)$ and $v_iv_n\notin E(G)$. But among the vertices that are visited before $v$ in $\sigma$ there is just a vertex $v_{n-1}$ that is adjacent to $v$. We have that $v_{n-1}v_n\in E(G)$, so the condition of~\cref{thm:ldfs-ordering-characterization} is not fulfilled, and $\sigma $ is not a LexDFS ordering of $G$.  
\end{proof}
It turns out that the equivalence between DFS and LexDFS in a graph $G$ implies that $G$ is a $\{$pan, diamond$\}$-free graph. In the following lemma we show that this is also sufficient. 

\begin{lemma}
If a graph $G$ does not contain a diamond, or a pan as an induced subgraph, then DFS and LexDFS are equivalent in $G$.
\end{lemma}
\begin{proof}

Let ${\cal G}$ be a class of $\{$diamond, pan$\}$-free graphs. We want to prove that DFS and LexDFS are equivalent in $G$. Assume for contradiction this is not the case, and let $G$ be a graph and $\sigma $ an ordering of $G$ that is DFS but not LexDFS ordering. 

Since $\sigma$ is a DFS ordering, it satisfies the characterization given in~\cref{thm:dfs-ordering-characterization}: if $ a <_\sigma b <_\sigma c$ and $ac \in E$ and $ab \notin E$, then there exists a vertex $d$ such that $a<_\sigma d <_\sigma b$ and $db \in E$.
From~\cref{thm:ldfs-ordering-characterization} 
it follows that there exist vertices $a,b,c$ in $G$ such that $a<_\sigma b<_\sigma c$, $ab\notin E(G), ac\in E(G)$ and for all vertices $d$ satisfying $a<_\sigma d<_\sigma b$ it holds that either $dc\in E(G)$, or $db\notin E(G)$.

Let $a<_\sigma b<_\sigma c$ be leftmost vertices that don't satisfy the characterization of LexDFS ordering $\sigma$ given in~\cref{thm:ldfs-ordering-characterization}. We know that $\sigma$ is DFS ordering of $G$, so there exists a vertex $d_1$ such that $a<_\sigma d_1<_\sigma b$ and $d_1b\in E(G)$. Then it follows that $d_1c\in E(G)$. 
Also, we have that $ad_1\notin E(G)$ and $bc\notin E(G)$, since otherwise we get a pan or a diamond. 

Consider now the vertices $a,d_1,c$. It holds that $a<_\sigma d_1<_\sigma c$, with $ac\in E(G)$ and $ad_1\notin E(G)$. These vertices satisfy the LexDFS ordering characterization, so there exists a vertex $d_2$ such that $a<_\sigma d_2<_\sigma d_1$, and $d_2d_1\in E(G)$, $d_2c\notin E(G)$. If $d_2b\in E(G)$, then the vertices $d_2,d_1,b,c$ form a $3$-pan. If $ad_2\in E(G)$, then the vertices $a,d_2,d_1,b,c$ form a $4$-pan. Hence, it follows that $ad_2\notin E(G)$ and $bd_2\notin E(G)$. Now we can continue this process by considering the vertices $a,d_2,c$ and apply the characterization of the LexDFS ordering. Let $d_1,d_2,\dots, d_k$ be a sequence of vertices defined in the following way: given a triple $a<_\sigma d_i<_\sigma c$ such that $ad_i\notin E(G)$, $ac\in E(G)$, $d_{i+1}$ is a vertex satisfying the conditions: $a<_\sigma d_{i+1}<_\sigma d_i$, $d_{i+1}d_i\in E(G)$, and $d_{i+1}c\notin E(G)$. Let $k$ be the maximum number of such vertices. We know that the number of vertices between $a$ and $c$ is finite, so $k$ is a finite number. We show the following claims:
\begin{enumerate}[nosep, label = \roman*)]
\item $d_id_{i+1}\in E(G)$, for all $i\in \{2,\dots, k\}$ - this is true by definition of vertices $d_i$, 
\item $d_{i+1}c\notin E(G)$, for all $i\in \{2,\dots, k\}$ - this is true by definition of vertices $d_i$,
\item $d_{i}a\notin E(G)$, for all $i\in \{1,2,\dots, k-1\}$ - this is true by definition of vertices $d_i$,
\item $d_ka\in E(G)$ - if this would not be true, then we can continue the process, and $d_k$ is not the last vertex in this sequence
\item $d_id_j\notin E(G)$, for all $i,j\in \{1,\dots, k\}$, such that $|i-j|\ge 2$ - Assume the opposite: let $d_i$ and $d_j$ be adjacent vertices with $|i-j|\ge 2$, such that $|i-j|$ is minimum and among all pairs $i,j$ satisfying this minimality condition, let $i,j$ be the smallest possible (equivalently, the right-most in the ordering $\sigma $). Without loss of generality we may assume that $j<i$. From the minimality of $|i-j|$ it follows that the vertices $d_j,d_{j+1}, \dots, d_{i}$ form an induced cycle in $G$. If $j=1$, then the vertices $\{d_j,d_{j+1}, \dots, d_{i}, c\}$ form a pan in $G$. Similarly, if $i=k$, then the vertices $\{d_j,d_{j+1}, \dots, d_{i}, a\}$ form a pan in $G$. It follows that $j>1$ and $i<k$. Consider now the vertex $d_{j-1}$. By the way we chose $i$ and $j$ it follows that $d_{j-1}d_\ell\notin E(G)$ for all $\ell\in\{d_{j+1},\dots, d_{i-1}\}$.
If $d_{j-1}d_i\notin E(G)$, then the vertices $\{d_{j-1},d_j,d_{j+1}, \dots, d_{i}\}$ form a pan in $G$. 
If $d_{j-1}d_i\in E(G)$, we consider two cases. 
First, if $i-j=2$, then the vertices $\{d_i, d_{i-1}, d_{i-2}=d_j, d_{i-3}=d_{j-1}\}$ form a diamond. 
Second, if $i-j>2$, then the vertices $\{d_{j-1}, d_j, d_i,d_{i-1}\}$ form a $3$-pan. 
In both cases we get a contradiction with the definition of $G$, meaning that such an edge $d_id_j$ cannot exists in $G$. 
\item $d_ib\notin E(G)$, for all $i\in\{2,\dots, k\}$ - Assume for contradiction that $j$ is a minimal value in $\{2,\dots, k\}$ such that $d_jb\in E(G)$. Then the vertices $\{d_1,\dots, d_j, b,c \}$ form a pan in $G$; a contradiction. 
\end{enumerate}

Consider now the vertices $\{d_1,\dots, d_k, a, c, b\}$. From the above claims it follows that they form a pan, where $b$ is a vertex of degree one. This is a contradiction with the definition of $G$. It follows that the vertices $d_1,\dots, d_k$ defined as above cannot exist, so $\sigma $ is a LexDFS ordering of $G$, as we wanted to show. 
\end{proof}
From the statements above the proof of main claim of this section follows immediately. 

\begin{corollary}\label{cor:(L)DFS}
For any graph $G$, the following is equivalent: 
\begin{enumerate}[nosep, label = \roman*)]
    \item Every DFS ordering of $G$ is a LexDFS ordering of $G$.
    \item Graph $G$ is $\{$pan, diamond$\}$-free.
\end{enumerate}
\end{corollary}

\subsection{Generic graph search and Maximal Neighbourhood Search}
\begin{lemma}\label{lem:MNS-cycles-forests-cliques}
In the following graph classes every graph search is an MNS ordering. 
\begin{enumerate}[nosep, label = \roman*)]
\item trees
\item cycles
\item complete graphs
\item complete bipartite graphs
\end{enumerate}
\end{lemma}

\begin{proof}
We prove each statement separately.
\begin{enumerate}[nosep, label = \roman*)]
\item Let $G$ be a tree, and fix any generic vertex-ordering $\alpha$. Since every non-starting vertex must have an $\alpha$-smaller vertex, the only way to violate MNS order paradigm would be to select a candidate with label which is a strict subset of a label by another candidate. 
Since in case of trees all candidates at all steps have the label of length exactly one, such violation cannot happen. 
\item In case of cycles, at every non-last step we are in a similar situation as in the case of trees -- every candidate vertex contains a label of length one, and is hence safe to choose. 
The only exception to this is the last vertex which will have a label of length two. Since it is the only remaining vertex, this will not violate MNS paradigm as well.
\item If $G$ is complete, then every ordering of vertices is equivalent, hence it is always generic search, as well as MNS search order.
\item Suppose $G$ is a complete bipartite graph on bipartitions $A,B$, and let $\alpha$ be any generic vertex-ordering. Furthermore let $a,b,c$ be arbitrary vertices such that $a<_alpha b<_alpha c$, and $ac\in E(G)$ while $ab\notin E(G)$, and wlog. assume $a\in A$. To satisfy \cref{thm:mns-ordering-characterization} it is enough to find a vertex $d<_\alpha b$ such that $db\in E(G)$ and $dc]\notin E(G)$. This is clearly true if $a$ is the starting vertex of $\alpha$, as in this case we fix $d$ to be its immediate successor and observe that $d,c\in B$ while $b\in A$.

But in the other case, when $a$ is not the start of $\alpha$, then set $d$ to be any of its neighbors such that $d<_\alpha a$. Such a neighbor exists as $\alpha$ is a generic search order. Again observe that $d,c\in B$ while $b\in A$, which concludes the proof of the claim.
\end{enumerate}
\end{proof}

\begin{lemma}
If a graph $G$ contains a pan, or a diamond as an induced subgraph, 
then there is a generic ordering of $G$ that is not a MNS ordering of $G$ (that is, generic search and MNS are not equivalent in $G$).
\end{lemma}

\begin{proof}
The claim can be easily justified by giving a prefix of an order $\sigma$ that is a search order and not an MNS order of a graph containing a pan or a diamond.
First consider the case when $G$ contains a diamond as an induced subgraph. Using the same notation as in~\cref{fig:3BFSnotLBFS} (right) we can define the search ordering $\sigma_2$ of $G$ starting in $c$, with first four vertices in $\sigma_2$ being $b,c,d,a$, in that order. It is clear that vertices $(b,d,a)$ violate the characterisation from \cref{thm:mns-ordering-characterization}.

Similarly, if $G$ contains an induced pan on vertices $v_0,v_1,\dots,v_k$ where $v_0$ and $v_2$ are of degrees 1 and 3, respectively. Now construct a search order which starts with 
$$
(v_2,v_3,\dots,v_k,v_0,v_1,\dots).
$$

Again, it is clear that the triplet $(b,d,a)$ violates the MNS search paradigm, which concludes the proof of the claim.
\end{proof}

From the statements above the proof of main claim of this section follows immediately. 

\begin{corollary}\label{cor:g-MNS}
For any graph $G$, the following is equivalent: 
\begin{enumerate}[nosep, label = \roman*)]
    \item Every graph search ordering of $G$ is a MNS ordering of $G$.
    \item Graph $G$ is $\{$pan, diamond$\}$-free.
\end{enumerate}
\end{corollary}

The above corollary, together with \cref{cor:(L)BFS,cor:(L)DFS} give the proof of \cref{thm:(L)DFS}.

\section{Proof of \cref{thm:thmMNSldfsLBFS}}\label{sec:P4-C4}

In this section we prove \cref{thm:thmMNSldfsLBFS}, that is, we characterize graphs for which it holds that every MNS ordering is also a LexBFS ordering, and graphs for which it holds that every MNS ordering is also a LexDFS ordering. As stated in \cref{thm:thmMNSldfsLBFS} it turns out that in both cases the same graphs are forbidden as induced subgraphs, as the following lemmas show. 
\begin{theorem}
If every MNS ordering of $G$ is also a LexBFS ordering of $G$, then $G$ is a $\{P_4,C_4\}$-free graph.
\end{theorem}
\begin{proof}
Let $G$ be a graph in which every MNS ordering is LexBFS ordering. Assume for contradiction that $G$ is not $\{P_4,C_4\}$-free graph.

Assume first that $G$ contains an induced $P_4$, and let $v_1,v_2,v_3,v_4$ be vertices of $P_4$. Let $\sigma$ be a MNS ordering of vertices in $G$ with $\sigma(1)=v_2$. Then any neighbor of $v_2$ can be selected next, so let $\sigma(2)=v_3$.  Then the label of vertex $v_4$ contains a vertex $v_3$, while a label of vertex $v_1$ does not contain it, meaning that the label of a vertex $v_4$ will never be a proper subset of a label of a vertex $v_1$, and we can select vertex $v_4$ before vertex $v_1$ in $\sigma$. At the same time once the vertices $v_2$ and $v_3$ are selected, from the definition of LexBFS it follows that all the neighbors of $v_1$ must be selected before its non-neighbors, so if $\sigma$ is a LexBFS ordering of $G$, it must be that $v_1<_\sigma v_4$; a contradiction.
If $G$ contains an induced $C_4$, the same reasoning holds, so we get a contradiction in any case and it follows that $G$ is a $\{P_4,C_4\}$-free graph.
\end{proof}

\begin{lemma}
If every MNS ordering of $G$ is also a LexDFS ordering of $G$, then $G$ is a $\{P_4,C_4\}$-free graph.
\end{lemma}
\begin{proof}
Let $G$ be a graph in which every MNS ordering is LexDFS ordering. Assume for contradiction that $G$ is not $\{P_4,C_4\}$-free graph.

Assume first that $G$ contains an induced $P_4$, and let $v_1,v_2,v_3,v_4$ be vertices of $P_4$. Let $\sigma$ be a MNS ordering of vertices in $G$ with $\sigma(1)=v_2$. Then any neighbor of $v_2$ can be selected next, so let $\sigma(2)=v_3$.  Then a label of vertex $v_4$ contains a vertex $v_3$, while a label of vertex $v_1$ does not contain it, meaning that the label of a vertex $v_4$ will never be a proper subset of a label of a vertex $v_1$, and we can select vertex $v_4$ before vertex $v_1$ in $\sigma$. At the same time once the vertices $v_2$ and $v_3$ are selected, from the definition of LexBFS it follows that all the neighbors of $v_1$ must be selected before its non-neighbors, so if $\sigma$ is a LexBFS ordering of $G$, it must be that $v_1<_\sigma v_4$; a contradiction.
If $G$ contains an induced $C_4$, the same reasoning holds, so we get a contradiction in any case and it follows that $G$ is a $\{P_4,C_4\}$-free graph.
\end{proof}

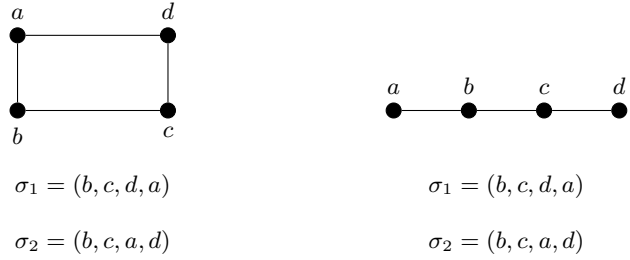
\begin{figure}[H]
\centering
\begin{tikzpicture}[vertex/.style={inner sep=2pt,draw, fill,circle}, ]
\begin{scope}
\node[vertex, label=below:$ b $] (b) at (0,0) {};
\node[vertex, label=below:$ c $] (c) at (2,0) {};
\node[vertex, label=above:$ a $] (a) at (0,1) {};
\node[vertex, label=above:$ d $] (d) at (2,1) {};
\draw[] (a) -- (b) -- (c) -- (d) -- (a);
\node[] (s1) at (1,-1) {$\sigma_1=(b,c,d,a)$};
\node[] (s2) at (1,-1.75) {$\sigma_2=(b,c,a,d)$};

\end{scope}
\begin{scope}[xshift=5cm]
\node[vertex, label=above:$ a $] (a) at (0,0) {};
\node[vertex, label=above:$ b $] (b) at (1,0) {};
\node[vertex, label=above:$ c $] (c) at (2,0) {};
\node[vertex, label=above:$ d $] (d) at (3,0) {};
\draw[] (a) -- (b) -- (c) -- (d);
\node[] (s1) at (1.5,-1) {$\sigma_1=(b,c,d,a)$};
\node[] (s2) at (1.5,-1.75) {$\sigma_2=(b,c,a,d)$};
\end{scope}
\end{tikzpicture}
\caption[MNS orderigns that are not LexBFS or LexDFS.]{A cycle (left) and a path (right) on $4$ vertices. $\sigma_1$ is a MNS ordering that is not a LexBFS ordering. $\sigma_2$ is a MNS ordering that is not a LexDFS ordering.}
\label{fig:forbiddengraphs-p4c4}
\end{figure}

It follows that given a graph $G$ satisfying the property that every MNS ordering is a LexBFS (resp., LexDFS) ordering, it must be true that $G$ is $\{P_4,C_4\}$-free graph. It turns out that this is also sufficient condition - in a $\{P_4,C_4\}$-free graphs every MNS ordering is also a LexBFS ordering and a LexDFS ordering. We prove these claims in the following two theorems. Observe that $\{P_4,C_4\}$-free graphs are also known as trivially-perfect graphs, and can be obtained from the $1$-vertex graphs using the operations of disjoint union and addition of universal vertices~\cite{golumbic1978trivially}.

\begin{lemma}\label{lem:p4c4-lexbfs}
Let $G$ be a $\{P_4,C_4\}$-free graph.
Then every MNS ordering of $G$ is also a LexBFS ordering of $G$. 
\end{lemma}
\begin{proof}

Let $G$ be a $\{P_4,C_4\}$-free graph, and assume for contradiction that there is an ordering $\sigma$ of vertices in $G$ that is a MNS ordering of $G$ and not a LexBFS ordering of $G$. From~\cref{thm:lbfs-ordering-characterization} we know that there exist vertices $a,b,c$ in $G$ such that $a<_\sigma
b<_\sigma c$ and $ac\in E(G)$, $ab\notin E(G)$, and for every $d<_\sigma a$ it holds that either $db\notin E(G)$ or $dc\in E(G)$. Let $a,b,c$ be the left-most such triple (that is, for any other triple $a'<_\sigma
b'<_\sigma c'$ and $a'c'\in E(G)$, $a'b'\notin E(G)$, with $\sigma(a')+\sigma(b')+\sigma(c')< \sigma(a)+\sigma(b)+\sigma(c)$ the~\cref{thm:lbfs-ordering-characterization} is satisfied). 

We know that $\sigma$ is MNS ordering, so by \cref{thm:mns-ordering-characterization} it follows that there exists a vertex $ d<_\sigma b$ in $G$ such that $db\in E(G)$ and $dc\notin E(G)$. It cannot be that $d<_\sigma a$, so it follows that have that $a<_\sigma d<_\sigma b$. If $ad\in E(G)$, or $bc\in E(G)$, then the vertices $\{a,b,c,d\}$ induce either a $P_4$, or a $C_4$ in $G$; a contradiction. It follows that $ad\notin E(G)$ and $bc\notin E(G)$. 

Now the vertices $a<_\sigma d<_\sigma c$ form a triple with $ac\in E(G)$ and $ad\notin E(G)$, so they must satisfy the \cref{thm:lbfs-ordering-characterization} and there exists a vertex $d_1<_\sigma a$ such that $d_1d\in E(G)$ and $d_1c\notin E(G)$. Moreover, it follows that $d_1a\notin E(G)$, for otherwise the vertices $\{d_1,a,d,c\}$ form a $P_4$ in $G$. 

Consider now the vertices $d_1<_\sigma a<_\sigma d$. They form a triple satisfying $d_1d\in E(G)$ and $d_1a\notin E(G)$, so by \cref{thm:lbfs-ordering-characterization} there exists a vertex $d_2<_\sigma d_1$ such that $d_2a\in E(G)$ and $d_2d\notin E(G)$. If $d_2d_1\in E(G)$, then the vertices $\{d_2,d_1,a,d\}$ form a $P_4$, a contradiction. 
We can continue the same process and apply \cref{thm:lbfs-ordering-characterization} on vertices $d_2,d_1,a$ in order to obtain a vertex $d_3$, and then apply the same process on vertices $d_i, d_{i-1}, d_{i-2}$ to obtain vertices $d_{i+1}$, for $i\ge 3$, as in \cref{thm:lbfs-ordering-characterization}.
Since a graph $G$ is finite, in this process we get the vertices $d_1,\dots, d_k$, for some finite number $k$. Let $k$ be the length of a maximal sequence of such vertices. It will be true that $d_i<_\sigma d_{i-1}$ for all $i\ge 2$, and 
\begin{equation}
\label{eqn:edges}
d_{i+2}d_i\in E(G)\, {\rm and }\, d_{i+3}d_i\notin E(G),
\end{equation} for all $i\ge 1$.

We prove the following claim inductively. 

\noindent\textbf{Claim :} $d_id_{i-1}\notin E(G)$, for $i\in \{2, \dots, k\}$

We know that $d_2d_1\notin E(G)$, so the inductive basis holds trivially. Assume now that for all $i\le j$ we have that $d_id_{i-1}\notin E(G)$. Let $i=j+1$. If $d_{j+1}d_j\in E(G)$, then the vertices $\{d_{j+1},d_j,d_{j-1},d_{j-2}\}$ induce a $P_4$ in $G$. This is true since $d_jd_{j-1}\notin E(G)$, $d_{j-1}\notin E(G)$ by inductive hypothesis, while other edges and non-edges follow from \ref{eqn:edges}. A contradiction with definition of $G$, so the claim follows.
 
It follows that vertices $d_k<_\sigma d_{k-1}<_\sigma d_{k-2}$ satisfy that $d_kd_{k-2}\in E(G)$ and $d_kd_{k-1}\notin E(G)$, so by \cref{thm:lbfs-ordering-characterization} there exists a vertex $d_{k+1}$ and $k$ is not maximal; a contradiction. 

\end{proof}

\begin{lemma}\label{lem:p4c4-lexdfs}
Let $G$ be a $\{P_4,C_4\}$-free graph.
Then every MNS ordering of $G$ is also a LexDFS ordering of $G$. 
\end{lemma}
\begin{proof}

Let $G$ be a $\{P_4,C_4\}$-free graph, and assume for contradiction that there is an ordering $\sigma$ of vertices in $G$ that is a MNS ordering of $G$ and not a LexDFS ordering of $G$. From~\cref{thm:ldfs-ordering-characterization} we know that there exist vertices $a,b,c$ in $G$ such that $a<_\sigma
b<_\sigma c$ and $ac\in E(G)$, $ab\notin E(G)$, and for every $a<_\sigma d<_\sigma b$ it holds that either $db\notin E(G)$ or $dc\in E(G)$. Let $a,b,c$ be the left-most such triple (that is, for any other triple $a'<_\sigma
b'<_\sigma c'$ and $a'c'\in E(G)$, $a'b'\notin E(G)$, with $\sigma(a')+\sigma(b')+\sigma(c')< \sigma(a)+\sigma(b)+\sigma(c)$ the~\cref{thm:ldfs-ordering-characterization} is satisfied). 

We know that $\sigma$ is MNS ordering, so by \cref{thm:mns-ordering-characterization} it follows that there exists a vertex $ d<_\sigma b$ in $G$ such that $db\in E(G)$ and $dc\notin E(G)$. It cannot be that $a<_\sigma d<_\sigma b$, so it follows that have that $d<_\sigma a$. If $ad\in E(G)$, or $bc\in E(G)$, then the vertices $\{a,b,c,d\}$ induce either a $P_4$, or a $C_4$ in $G$; a contradiction. It follows that $ad\notin E(G)$ and $bc\notin E(G)$. 

Now the vertices $d<_\sigma a<_\sigma b$ form a triple with $db\in E(G)$ and $da\notin E(G)$, so they must satisfy the \cref{thm:ldfs-ordering-characterization} and there exists a vertex $d<_\sigma d_1<_\sigma a$ such that $d_1a\in E(G)$ and $d_1b\notin E(G)$. Moreover, it follows that $dd_1\notin E(G)$, for otherwise the vertices $\{d,d_1,a,b\}$ form a $P_4$ in $G$. 

Consider now the vertices $d<_\sigma d_1<_\sigma b$. They form a triple satisfying $db\in E(G)$ and $dd_1\notin E(G)$, so by \cref{thm:ldfs-ordering-characterization} there exists a vertex $d<_\sigma d_2<_\sigma d_1$ such that $d_2d_1\in E(G)$ and $d_2b\notin E(G)$. If $dd_2\in E(G)$, then the vertices $\{d,d_2,d_1,b\}$ form a $P_4$, a contradiction. 
We can continue the same process and apply \cref{thm:ldfs-ordering-characterization} on vertices $d,d_2,b$ in order to obtain a vertex $d_3$, and then apply the same process on vertices $d,d_i,b$ to obtain vertices $d_{i+1}$, for $i\ge 3$, as in \cref{thm:ldfs-ordering-characterization}.
Since a graph $G$ is finite, in this process we get the vertices $d_1,\dots, d_k$, for some finite number $k$.

Let $k$ be the length of a maximal sequence of such vertices. It will be true that $d<_\sigma d_i<_\sigma d_{i-1}$ for all $i\ge 2$, and 
\begin{equation}
\label{eqn:edges2}
d_{i+1}d_i\in E(G)\, {\rm and }\, d_ib\notin E(G),
\end{equation} for all $i\ge 1$.

We prove the following claim inductively. 

\noindent\textbf{Claim :} $dd_i\notin E(G)$, for $i\in \{1,2, \dots, k\}$

We know that $dd_1\notin E(G)$, so the inductive basis holds trivially. Assume now that for all $i\le j$ we have that $dd_{i}\notin E(G)$. Let $i=j+1$. If $dd_{j+1}\in E(G)$, then the vertices $\{b,d,d_{j+1},d_j\}$ induce a $P_4$ in $G$. This is true since $dd_{j}\notin E(G)$ by inductive hypothesis, while other edges and non-edges follow from \ref{eqn:edges2}. A contradiction with definition of $G$, so the claim follows.
 
It follows that vertices $d<_\sigma d_k<_\sigma b$ satisfy that $db\in E(G)$ and $dd_k\notin E(G)$, so by \cref{thm:ldfs-ordering-characterization} there exists a vertex $d_{k+1}$ such that $d<_\sigma d_{k+1}<_\sigma d_k$ and $k$ is not maximal; a contradiction. 
\end{proof}

From \cref{lem:p4c4-lexbfs,lem:p4c4-lexdfs} the proof of main theorem of this section follows  immediately. 
\thmMNSldfsLBFS*

In other words, it follows that MNS and LexBFS are equivalent in $G$ if and only if $G$ is a $\{P_4,C_4\}$-free graph, and similarly, MNS and LexDFS are equivalent in $G$ if and only if $G$ is a $\{P_4,C_4\}$-free graph.

\fi

\end{document}